\documentclass[times,doublespace]{simauth-del}
\usepackage{multirow}
\newtheorem{theorem}{Theorem}
\newtheorem{proof}{Proof}
\begin{document}
\runninghead{J.~Kong, S.~Wang and G.~Wahba}

\title{Using distance covariance for improved variable selection with application to genetic risk models}

\author{Jing Kong\affil{a}\corrauth,
Sijian Wang\affil{b} and Grace Wahba\affil{c}}

\address{\affilnum{a}Department of Statistics, 1300 University Avenue, Madison, WI, 53706, US\\
\affilnum{b}Department of Statistics and Department of Biostatistics $\&$ Medical Informatics, 1300 University Avenue, Madison, WI, 53706, US\\
\affilnum{c}Department of Statistics, Department of Computer Sciences and Department of Biostatistics $\&$ Medical Informatics,1300 University Avenue, Madison, WI, 53706, US}

\corraddr{E-mail: kong@stat.wisc.edu}

\cgs{J.K. and G.W. are supported by National Science Foundation (NSF)
  Grant DMS1308877 and  National Institutes of Health (NIH) Grant EY09946. S.W is supported by NIH Grant 1R01HG007377.}

\begin{abstract}
Variable selection is of increasing importance to address the difficulties of high dimensionality in many scientific areas. In this paper, we demonstrate a property for distance covariance, which is incorporated in a novel feature screening procedure together with the use of distance correlation. The approach makes no distributional assumptions for the variables and does not require the specification of a regression model, and hence is especially attractive in variable selection given an enormous number of candidate attributes without much information about the true model with the response. The method is applied to two genetic risk problems, where issues including uncertainty of variable selection via cross validation, subgroup of hard-to-classify cases and the application of a reject option are discussed.
\end{abstract}

\maketitle

\section{Introduction}
The idea of feature screening came along as high dimensional data were collected in modern technology. It was aimed at dealing with the challenges of computational expediency, statistical accuracy and algorithmic stability due to high dimensionality. Fan and Lv proposed the sure independence screening (SIS) \cite{fan2008sure} and showed that the Pearson correlation ranking procedure possessed a sure screening property for linear regression with Gaussian predictors and responses. A new feature screening procedure for high dimensional data based on distance correlation \cite{szekely2007measuring}, named DC-SIS, was presented in \cite{li2012feature}. DC-SIS retained the sure screening property of the SIS, and additionally possessed new advantages of handling grouped predictors and multivariate responses by using distance correlation. Moreover, since distance correlation was applicable to arbitrary distributions, DC-SIS could also be used for screening features without specifying a regression model between the response and the predictors, and thus was robust to model mis-specification.

However, both SIS and DC-SIS relied on a user-specified model size $d$ which decided the number of predictors being selected. Let the sample size be $n$, $d$ was chosen to be multipliers of the integer part of $n/\log n$ in \cite{fan2008sure} and \cite{li2012feature} which did not depend on any other characteristics of the data. As pointed out by a referee of \cite{li2012feature}, the choice of $d$ was of great importance in practical implementation and might influence the screening results significantly. Our study is aimed at fixing this shortcoming by including an automatic stopping criteria for DC-SIS based on the property of distance covariance.

The screening procedures may fail if a feature is marginally uncorrelated, but jointly correlated with the response, or in the reverse situation where a feature is jointly uncorrelated but has higher marginal correlation than some important features. An iterative SIS was proposed in \cite{fan2008sure} to fix this problem. Current research interest involves dealing with this drawback but this work is not related to this quest.

We demonstrate our improved method through two real examples. The small round blue cell tumors (SRBCT) data were relatively easy to classify and had been studied extensively. The Cancer Genome Atlas (TCGA) ovarian cancer data, however, were much more challenging due to the large number of genes and limited sample size. The target was to identify the important genes that contribute to the sensitive or resistant status after receiving a particular chemotherapy treatment. A substantial fraction of the population was difficult to classify and a ``withholding decision" option is allowed in the support vector machine with reject option model to adapt to this fact. A multiple cross validation is used to quantify uncertainty given a humongous number of candidates, and we see a commonly observed dilemma that different variables are selected by using different subsets of the data. Comparison between the results from the original data and those from the data obtained by randomly permuting the response provide further justification on our conclusions. Furthermore, the multiple cross validation on the permuted data discloses the existence of spuriously correlated variables in high dimensional data and thus failure of variable selection and model building based on training data.

\section*{Some Preliminaries}
\subsection*{Distance correlation}
\cite{szekely2007measuring} proposed distance correlation as a measurement of dependence between two random vectors. The method has been successfully applied to various problem, see \cite{kong2012using} for example. To be specific, the authors defined the distance covariance between $X\in\mathbb{R}^p$ and $Y\in\mathbb{R}^q$ to be
\begin{displaymath}
\operatorname{V}^2(X,Y)= \frac{1}{c_p c_q}\int_{\mathbb{R}^{p+q}} \frac{\left| f_{X,Y}(s, t) - f_X(s)f_Y(t) \right|^2}{|s|_p^{1+p} |t|_q^{1+q}} dt\,ds
\end{displaymath}
where $f_{X, Y}(s, t), f_X(s)$, and $f_Y(t)$ are the characteristic functions of $(X, Y), X$, and $Y$ respectively, and $c_p$, $c_q$ are constants chosen to produce scale free and rotation invariant measure that doesn't go to zero for dependent variables. The idea is originated from the property that the joint characteristic function factorizes under independence of the two random vectors. This leads to the remarkable property that $V^2(X,Y)=0$ if and only if $X$ and $Y$ are independent.

The sample version of distance covariance and distance correlation involves pairwise distances. For a random sample $(X,Y)=\{(X_k, Y_k): k=1,...,n\}$ of $n$ i.i.d random vectors $(X,Y)$ from the joint distribution of random vectors $X$ in $\mathbb{R}^p$ and $Y$ in $\mathbb{R}^q$, the Euclidean distance matrices $(a_{ij})=(|X_i-X_j|_p)$ and $(b_{ij})=(|Y_i-Y_j|_q)$ with $i,j = 1,\ldots,n$ are computed. Define the double centering distance matrices
\begin{displaymath}
A_{ij}=a_{ij}-{\overline{a}_{i\cdot}}-{\overline{a}_{\cdot j}}+{\overline{a}_{\cdot\cdot}},\hspace{0.3cm}i,j=1,\ldots,n,
\end{displaymath}
where
\begin{displaymath}
{\overline{a}_{i\cdot}}=\frac{1}{n}\sum_{j=1}^na_{ij},\hspace{0.3cm}{\overline{a}_{\cdot j}}=\frac{1}{n}\sum_{i=1}^na_{ij},\hspace{0.3cm}{\overline{a}_{\cdot\cdot}}=\frac{1}{n^2}\sum_{i,j=1}^na_{ij},
\end{displaymath}
similarly for $B_{ij}=b_{ij}-{\overline{b}_{i\cdot}}-{\overline{b}_{\cdot j}}+{\overline{b}_{\cdot\cdot}},\hspace{0.3cm}i,j=1,\ldots,n$. Then, the sample distance covariance $V_n(X,Y)$ is defined by
\begin{displaymath}
V_n^2(X,Y)=\frac{1}{n^2}\sum_{i,j=1}^nA_{ij}B_{ij}.
\end{displaymath}
The sample distance correlation $R_n(X,Y)$ is defined by
\begin{equation}
R_n^2(X,Y)=\begin{cases}\displaystyle \frac{V_n^2(X,Y)}{\sqrt{V_n^2(X)V_n^2(Y)}}, & V_n^2(X)V_n^2(Y)>0; \\
0, &V_n^2(X)V_n^2(Y)=0,\end{cases} \notag
\end{equation}
where the sample distance variance is defined by
\begin{displaymath}
V_n^2(X)=V_n^2(X,X)=\frac{1}{n^2}\sum_{i,j=1}^nA_{ij}^2.
\end{displaymath}

\subsection*{Feature screening via distance correlation (DC-SIS)}
\cite{fan2008sure} proposed sure independence screening (SIS) procedure based on the Pearson correlation for feature selection. The distance correlation version of this technique (DC-SIS) was studied in \cite{li2012feature}. With a user-specific model size $d$, the variables whose distance correlations with the response ranking from 1st to $d$th in decreasing order were selected. The authors explored the theoretic properties of the DC-SIS and proved that the DC-SIS kept the desired sure screening property established in \cite{fan2008sure}. Moreover, due to the property of distance correlation, DC-SIS procedure was robust to model mis-specification, which was demonstrated in their simulations.

\section*{Improving DC-SIS using distance covariance}
\begin{theorem}Suppose random vectors $X,Z\in\mathbb{R}^p$ and $Y\in\mathbb{R}^q$, and assume $Z$ is independent of $(X,Y)$, then
\begin{equation}
V^2(X+Z,Y)\leq V^2(X,Y),
\end{equation}
where $V$ is the population distance variance defined in \cite{szekely2007measuring}.
\end{theorem}
\begin{proof}
\begin{align}
V^2(X+Z,Y)&= \parallel f_{X+Z,Y}(t,s) - f_{X+Z}(t)f_{Y}(s)\parallel^2 \notag\\
&=\frac{1}{c_pc_q}\int_{\mathcal{R}^{p+q}}\frac{1}{|t|_p^{1+p}|s|_q^{1+q}}|f_{X+Z,Y}(t,s) - f_{X+Z}(t)f_{Y}(s)|^2dtds. \notag
\end{align}
The following fact follows from the definition of characteristic function and independence assumption.
\begin{align}
&|f_{X+Z,Y}(t,s) - f_{X+Z}(t)f_{Y}(s)|^2 \notag\\
=&|Ee^{it^T(X+Z)+is^TY} - Ee^{it^T(X+Z)}Ee^{is^TY}|^2 \notag \\
=&|Ee^{it^TX+is^TY}Ee^{it^TZ} - Ee^{it^TX}Ee^{it^TZ}Ee^{is^TY}|^2\notag\\
=&|f_{X,Y}(t,s)f_Z(t) - f_{X}(t)f_Z(t)f_{Y}(s)|^2\notag\\
=&|f_Z(t)|^2|f_{X,Y}(t,s) - f_{X}(t)f_{Y}(s)|^2, \notag
\end{align}
Since $|f_Z(t)|\leq 1$ by the property of characteristic function\footnote{In \cite{kosorok2009discussion}, the author obtained equality here which is incorrect.}, we have
\begin{equation}
|f_{X+Z,Y}(t,s) - f_{X+Z}(t)f_{Y}(s)|^2 \leq |f_{X,Y}(t,s) - f_{X}(t)f_{Y}(s)|^2, \notag
\end{equation}
which implies
\begin{align}
V^2(X+Z,Y)&\leq\frac{1}{c_pc_q}\int_{\mathcal{R}^{p+q}}\frac{1}{|t|_p^{1+p}|s|_q^{1+q}}|f_{X,Y}(t,s) - f_{X}(t)f_{Y}(s)|^2dtds\notag\\
&=\parallel f_{X+Z,Y}(t,s) - f_{X+Z}(t)f_{Y}(s)\parallel^2 \notag\\
&=V^2(X,Y). \notag
\end{align}
\end{proof}

We know that if $E|X|_p<\infty, E|X+Z|_p<\infty$ and $E|Y|_p<\infty$, then almost surely
\begin{align}
\displaystyle \lim_{n\rightarrow\infty}V_n^2(X+Z,Y)&=V^2(X+Z,Y),\notag\\
\displaystyle \lim_{n\rightarrow\infty}V_n^2(X,Y)&=V^2(X,Y).\notag
\end{align}
Thus, for the sample distance covariance, if $n$ is large enough, we should have
\begin{equation}
V_n^2(X+Z,Y)\leq V_n^2(X,Y), \notag
\end{equation}
under the assumption of independence between $(X,Y)$ and $Z$.

In the case where $(X,Z)$ is of interest, which is the usual situation for variable selection setting, we could use the above theorem by incorporating degenerated random vectors as follows. Suppose $X\in\mathbb{R}^{p_1}$ and $Z\in\mathbb{R}^{p_2}$, then we augment $X$ and $Z$ to be $\tilde{X}=(X,0_{p_2})$ and $\tilde{Z}=(0_{p_1},Z)$ respectively. $\tilde{X}$ and $\tilde{Z}$ are therefore of the same dimension and $\tilde{X} + \tilde{Z} = (X,Z).$

We implemented the above theorem as a check for stopping for DC-SIS. For the original DC-SIS, it required a user-specified model size $d$, which was always chosen as multipliers of the integer part of $n/\log n$. For our improved screening procedure with distance correlation, we first ranked the importance of $x_i, i = 1,...,p$ using the marginal distance correlations with the response as DC-SIS did and initialized $\mathcal{S}$ as the singleton including the index of the top one variable. Instead of selecting the top $d$ variables, we kept adding variables to $x_\mathcal{S} = \{x_i:i\in\mathcal{S}\}$ according to the ordered list of variables until observing a decrease in the distance covariance between $x_\mathcal{S}$ and $y$. The procedure took the following steps and we denoted the procedure as DCOV method.
\begin{enumerate}
\item Calculate marginal distance correlations for $x_i, i = 1,...,p$ with the response.
\item Rank the variables in decreasing order of the distance correlations. Denote the ordered variables as $x_{(1)}, x_{(2)},...,x_{(p)}$. Start with $x_{\mathcal{S}} = \{x_{(1)}\}$.
\item For $i$ from $2$ to $p$, keep adding $x_{(i)}$ to $x_{\mathcal{S}}$ if $V_n^2(x_{\mathcal{S}},y)$ does not decrease. Stop otherwise.
\end{enumerate}

\section*{Real application on SRBCT data}
The small round blue cell tumors (SRBCTs) are 4 different childhood tumors named so because of their similar appearance on routine histology, which makes correct clinical diagnosis extremely challenging. However, accurate diagnosis is essential because the treatment options, responses to therapy and prognoses vary widely depending on the diagnosis. They include Ewing's family of tumors (EWS), neuroblastoma (NB), non-Hodgkin lymphoma (in our case Burkitt's lymphoma, BL) and rhabdomyosarcoma (RMS). The SRBCT data being published in \cite{khan2001classification} included the expression of 2308 genes measured on 63 samples (23 EWS, 8 BL, 12 NB and 20 RMS). This data are known as an easy-classified example and have been studied by many. \cite{lee2004multicategory} using the multicategory SVM is one of several methods that have excellent classification results on this data set. Hence, we focus more on the selected genes.

We applied our improved feature screening procedure on this dataset and compared our selection of genes with the 96 top genes reported in \cite{khan2001classification}. This is a multicategory classification and the genes were screened in a one-versus-rest fashion. Specifically, for each of the four different types of tumors, we generated a response indicator vector taking value of 0 if the sample came from the current interested type and 1 otherwise.
This allowed us to implement our screening procedure and obtained the genes which showed high distance correlation with the current type of tumor. The four groups of selected genes were combined as a whole collection of in total 176 genes. 47 genes turned out to be in common for the DCOV selection and the top 96 genes used in \cite{khan2001classification}.

We further examined the power of these two groups of genes in differentiating the 4 types of tumors by presenting the pairwise distances of the 63 samples
(\textit{Figure 1}). As shown in the plot, the samples were arranged in the order of EWS, BL, NB and RMS. The pairwise distances resulted from the two
selections of genes were scaled to maximum of 1 respectively so that they shared the same magnitude. Both groups of genes could tell the 4 types
of tumors apart. Compared with the 96 genes from \cite{khan2001classification}, however, the 176 DCOV selected genes show better distinguishability and clearer contrast over the 4 classes. Moreover, the right panel almost missed the samples labeled from 57
to 62 in the class of RMS but the 176 DCOV genes could recognize them with big differences between the in and out class pairwise distances. The dataset were known to be easy for classification and both sets of genes were able to classify the testing set of 20 samples perfectly via k-nearest neighbor method with $k=3$.


\begin{figure}
\begin{center}
  \includegraphics[width=18cm]{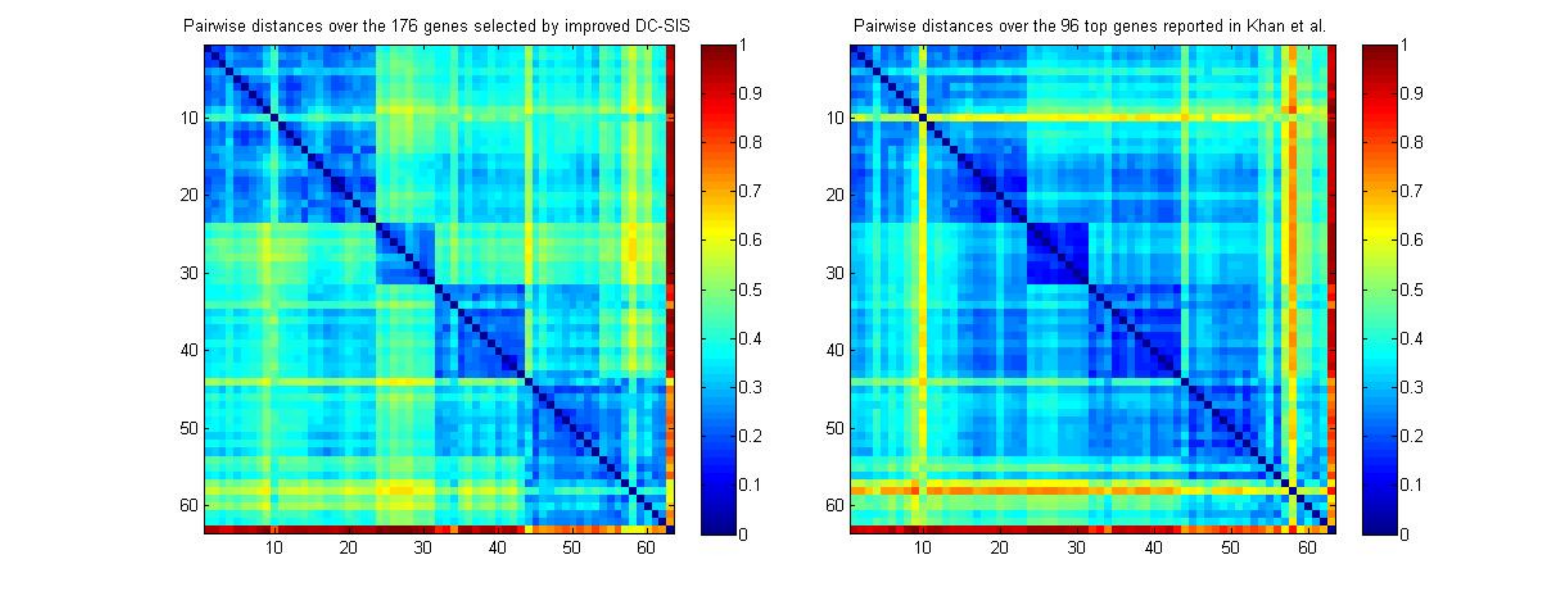}
  \caption{Comparison of pairwise distances between the two selections of genes. Left and right panel present the pairwise distances of the 63 samples over the improved DC-SIS selection of 176 genes and the 96 reported genes in \cite{khan2001classification} respectively.}
\end{center}
\end{figure}

\section*{Real application on TCGA ovarian cancer data}
\subsection*{Data description}
Ovarian cancer is the fifth-leading cause of cancer death among women in the United States; 22,240 new cases and 14,030 deaths were estimated to have occurred in 2013\cite{website,bell2011integrated}. The standard treatment for high-grade serous ovarian cancer is aggressive surgery followed by chemotherapy. Despite treatment, a vast majority of ovarian cancer patients eventually relapse and die of their disease with a major cause of chemotheraphy resistance \cite{selvanayagam2004prediction}. Identification and prediction of patients with chemoresistant thus become important for improving the outcome of ovarian cancer.

The Cancer Genome Atlas (TCGA) collected high-quality, high-dimensional, and multi-modal genetic data from women with ovarian cancer. There were 279 samples  with explicit chemostatus and gene expression (Affymetrix HT-HGU133a) data in the public set. among which 191 subjects were sensitive to chemotherapy and 88 were chemoresistant. Expression data for 12042 genes after log transformation are used for analysis. The issue is to explore whether there are genes whose expression pattern is strongly correlated with the response indicating chemotherapy status.

\subsection*{DCOV gene selection results based on all the observations}
Our feature screening procedure on the gene expression data for the 279 patients selected 82 genes, among which 5 were reported to be related to ovarian cancer in the literature. IGFBP5 ranked as the 5th is one of the six members of insulin-like growth factor-binding protein (IGFBP) family and is known to be important for cell growth control, induction of apoptosis and other IGF-stimulated signaling pathways. IGFBP5 expression is shown to prevent tumor growth and inhibited tumor vascularity in a xenograft model of human ovarian cancer and is suggested that IGFBP5 plays a role as tumor suppressor by inhibiting angiogenesis \cite{rho2008insulin}. GPR3, the 7th, is a member of a family of G-protein couple receptors whose activation of PKA and subsequent increase of cyclic AMP level promotes meiotic arrest in the oocyte \cite{mehlmann2004gs}. Mice deficient in GPR3 display premature ovarian aging and loss of fertility \cite{ledent2005premature}. MAPK4, the 18th, is a member of MAPK signaling pathway. MAPK signal transduction cascade is dysregulated in a majority of human tumors \cite{basu2009nanoparticle}. It is suggested playing an important role in molecular diagnostics and molecular therapeutics for lowgrade ovarian cancer \cite{bast2010personalizing}. FZD5 ranked as the 22th encodes Frizzled-5 protein, which is believed to be the receptor for the Wnt5A ligand \cite{thiele2011expression}. The Wnt5A ligand plays a context-dependent role in human cancers. It has been demonstrated that Wnt5a is expressed at significantly lower levels in human Epithelial ovarian cancer (EOC) cell lines and in primary human EOCs compared with either normal ovarian surface epithelium or fallopian tube epithelium \cite{bitler2011wnt5a}. FGF22, the 56th, is a member of Fibroblast Growth Factors (FGFs) family, whose members possess broad mitogenic and cell survival activities, and are involved in a variety of biological processes, including embryonic development, cell growth, morphogenesis, tissue repair, tumor growth and invasion. The inhibition of FGFR2, which is a member of this family, has been found to increase cisplatin sensitivity in ovarian cancer \cite{cole2010inhibition}.

39 pathways were found to be associated with the 82 genes, among which 3 pathways are known to be important for ovarian cancer. MAPK signaling pathway is suggested playing an important role in molecular diagnostics and molecular therapeutics for lowgrade ovarian cancer \cite{bast2010personalizing}. Wnt signaling pathway is best known for its role in tumorigenesis. \cite{bast2010personalizing} demonstrated the difference in Wnt signaling pathway between normal ovarian and cancer cell lines and between benign tissue and ovarian cancer. They also pointed out that those differences implicate that Wnt signaling leads to ovarian cancer development despite the fact that gene mutations are uncommon. \cite{yin2011genetic} suggested that genetic variants in TGF-$\beta$ signaling pathway are associated with ovarian cancer risk and may facilitate the identification of high-risk subgroups in the general population.

\subsection*{Support vector machine with reject option}
We estimated the probabilities of being chemosensitive or chemoresistant by a penalized Bernoulli likelihood main effect spline model using the \texttt{R} package \texttt{gss}\cite{gu2007gss}. Aside from the additive expression effects of the selected 82 genes, we also included two more covariates, namely the cancer grade and cancer stage of the subjects. Cancer grade is an indicator for grade 2 and grade 3. Cancer stage indicates whether the subject is in stages IIIC and IV or not. As shown in \textit{Figure 2}, the estimated probabilities have high density around small and large values for sensitive and resistant patients respectively, with overlapping in the middle values. This suggested that we were less confident about the chemostatus for the patients in the middle range and so we sought a principled approach which withholds decision for such cases.
\begin{figure}
\begin{center}
  \includegraphics[width=14cm]{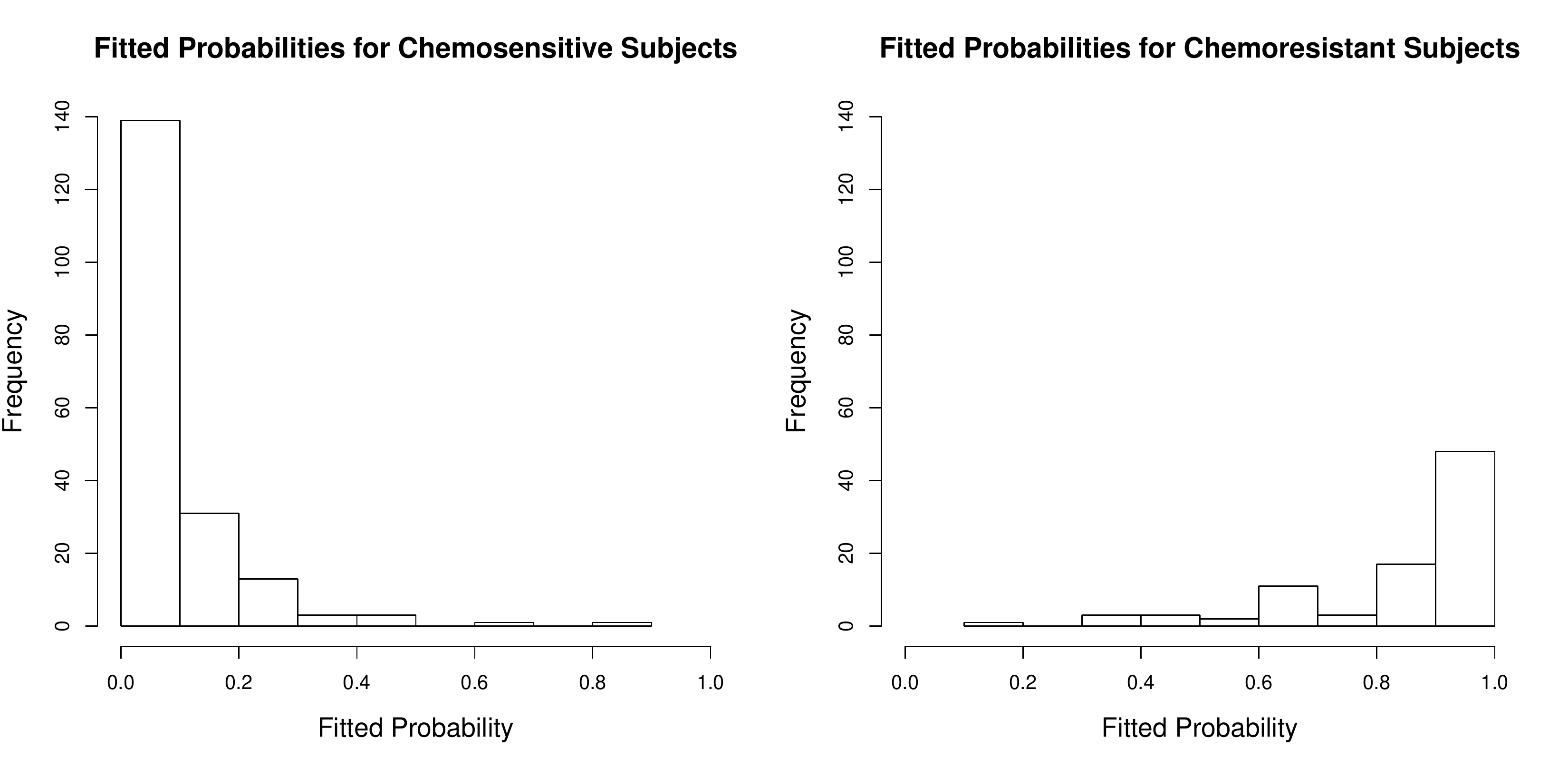}
  \caption{Fitted probabilities by penalized Bernoulli likelihood model with the 82 genes.}
\end{center}
\end{figure}

\cite{wegkamp2011support,bartlett2008classification} investigated the support vector machines with a built-in reject option for binary classification where the results of classification could be $-1, +1$ or withhold decision. Given a discriminant function $f: \mathcal{X}\rightarrow \mathbb{R}$, the method only reports $sgn(f(x))\in\{-1,1\}$ if $|f(x)| > \delta$ and withholds decision if $|f(x)|\leq \delta$. Suppose that the cost of making a wrong decision is 1 and that of rejecting to make a decision is $d\in[0,\frac{1}{2}]$, then an proper risk function is
\begin{displaymath}
L_{d,\delta}(f)=El_{d,\delta}(Yf(X)) = P\{Yf(X)<-\delta\}+dP(|Yf(X)|\leq\delta)
\end{displaymath}
with the discontinuous loss function
\begin{equation}
l_{d,\delta}(z) =
\begin{cases}
    1, & \hbox{if $z<-\delta$;} \\
    d, & \hbox{if $|z|\leq \delta$;} \\
    0, & \hbox{otherwise.}\notag
\end{cases}
\end{equation}
The classifier associated with the discriminant function
\begin{equation}
f^\ast_d(x) =
\begin{cases}
    -1, & \hbox{if $\eta(x)<d$;} \\
    0, & \hbox{if $d\leq \eta(x) \leq 1-d$;} \\
    +1, & \hbox{if $\eta(x)>1-d$,}\notag
\end{cases}
\end{equation}
with $\eta(x) = P\{Y=+1 | X=x\}$ minimizes the risk $L_{d,\delta}(f)$ with
\begin{displaymath}
L^\ast_d=L_{d,\delta}(f^\ast_d)=E\min\{\eta(X),1-\eta(X),d\}.
\end{displaymath}
To avoid working with the discontinuous loss $l_{d,\delta}$, \cite{wegkamp2011support,bartlett2008classification} proposed a convex surrogate loss, which is the generalized hinge loss,
\begin{equation}
\phi_d(z) =
\begin{cases}
    1-az, & \hbox{if $z<0$;} \\
    1-z & \hbox{if $0\leq z < 1$;} \\
    0, & \hbox{otherwise,}\notag
\end{cases}
\end{equation}
where $a = (1-d)/d\geq 1$. It followed that $f^\ast_d$ also minimizes the risk associated with $\phi_d$ over all measurable $f:\mathcal{X}\rightarrow\mathbb{R}$.

The discriminant functions $f$ took the form $f_{\lambda}(x)=\sum_{j=1}^M\lambda_jf_j(x)$ based on a set of known functions $f_j:\mathcal{X}\rightarrow\mathbb{R}$ and coefficients $\lambda_j\in\mathbb{R},1\leq j \leq M$. The coefficients were chosen to minimize the empirical risk
\begin{displaymath}
\hat{R}_{\phi}(f_\lambda) = \frac{1}{n}\sum_{i=1}^n\phi(Y_if_\lambda(X_i)).
\end{displaymath}
To reflect the preference for sparse solutions, which is desirable when $M$ is large compared to the sample size $n$, an $l_1$ type restriction $\|\lambda\|_1 =\sum_{j=1}^M|\lambda_j|$ was incorporated in \cite{wegkamp2011support} and $f_\lambda$ is estimated by $f_{\hat{\lambda}_r}$, where
\begin{equation}
\hat{\lambda}(r) = \textrm{arg}\min_{\lambda\in\mathbb{R}^M}\hat{R}_{\phi}(f_\lambda)+r\|\lambda\|_1
\end{equation}
and $r>0$ is a tuning parameter. We followed \cite{wegkamp2011support} to call this model support vector machines with reject option(SVM-R).

The authors in \cite{bartlett2008classification} also showed that the choice of $\delta = 1/2$ gives the optimal bound established by the excess risk of $\phi_d$ on the excess risk $L_{d,\delta}-L^\ast_d$ for any fixed $d\in [0,1/2)$ and measurable function $f$. For this reason, we fixed $\delta = 1/2$ for our practical use of the method. Furthermore, we took the set of known functions $f_j:\mathcal{X}\rightarrow\mathbb{R}$ to be linear functions of the log transformation on the 12024 genes. The optimization problem $(2)$ was formulated into a linear programming task and solved using \textsc{MATLAB}.

\textit{Figure 3} presents the 82 genes for the 279 subjects in groups according to the SVM-R classification results. The results correspond to the particular choice of $d = 1/4$ and $r = 4$ to illustrate the benefits of SVM-R. As shown in the plot, there is a big difference in the gene
expression between the subjects assigned to be resistant and sensitive. The behavior of the 82 genes for those without a certain decision tends to be in-between.

\begin{center}
\begin{figure}[h]
  \includegraphics[width=18cm]{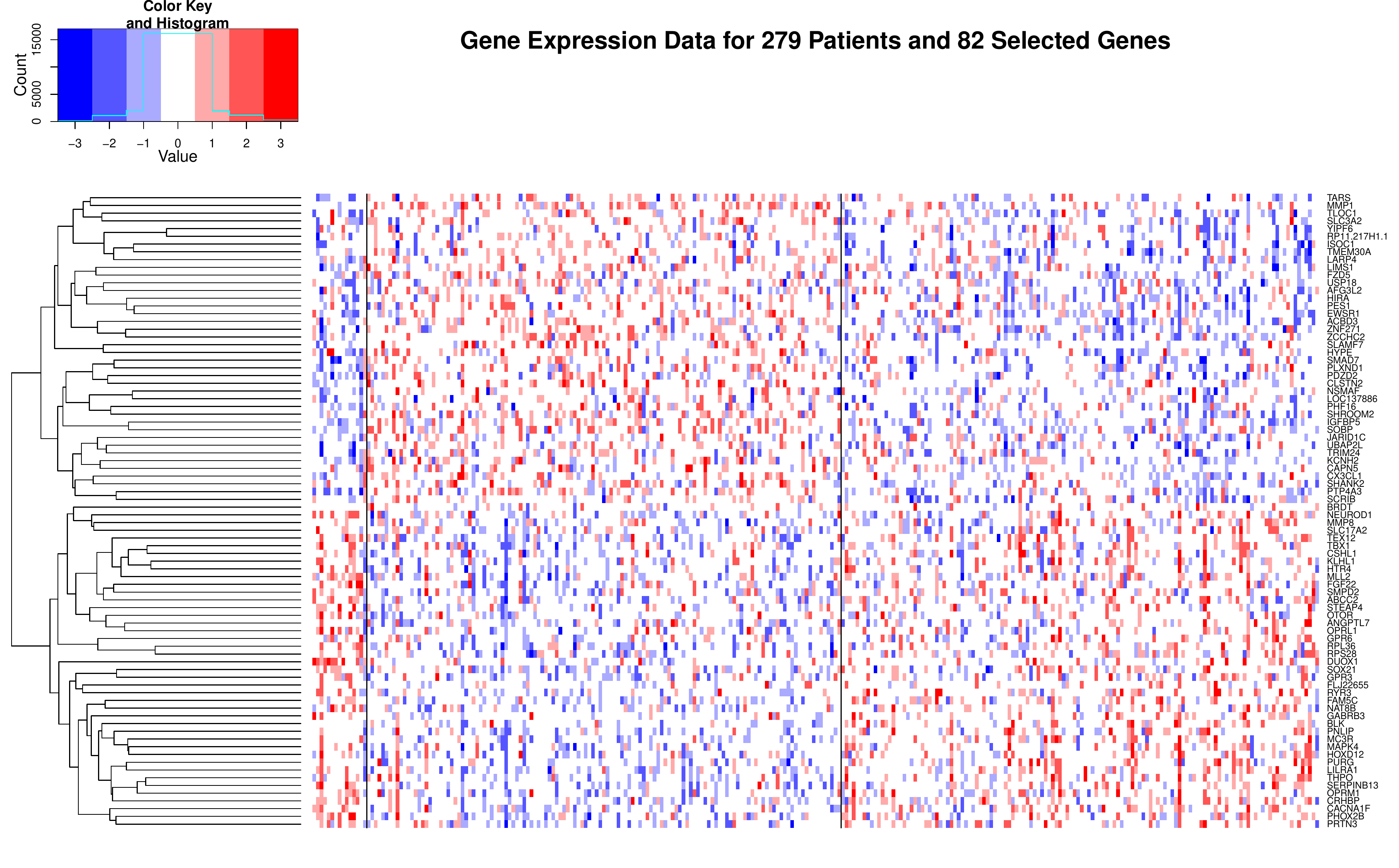}
  \caption{Gene expression data for the 82 selected genes and 279 subjects with SVM-R classification for $d = 1/4$ and $r = 4$. The subjects are grouped according to their assigned decisions by the SVM with a reject option. The left group involves 15 patients (1 sensitive and 14 resistant) classified to be resistant. The middle group is assigned to be sensitive and contains 123 sensitive and 8 resistant subjects. 67 sensitive patients and 66 resistant patients with a withhold decision are shown in the right group.}
\end{figure}
\end{center}

\subsection*{Five fold cross validation}
In order to choose the tuning parameter in SVM-R, we need to hold aside a tuning set before selecting the genes. Leaving out different observations leads to different gene selection results. Here we applied a five fold cross validation analysis to examine the variations of selections of genes and SVM-R model fitting results across different partitions of the dataset. The implementation followed the steps below.
\begin{enumerate}
\item Randomly partition the 279 subjects into 5 non-overlapping folds.
\item Select genes from the 12024 genes based on 4 folds as the training set.
\item Build SVM-R model with the selected genes and the two cancer status variables based on the training set.
\item Use the leaving-out fold as the tuning set to choose the tuning
  parameter for SVM-R with mean $l-$loss, defined below, as the criteria.
\item Repeat $2.-4.$ for the 5 folds.
\end{enumerate}

The $l-$loss for a subject is $1$ if a misclassification occurs, $d$
if a withholding decision is made and $0$ otherwise. The mean $l-$loss
is the average over the $l-$losses for all the subjects in a given set
of data. We looked for tuning parameter values minimizing the mean $l-$loss.

The above procedure produced 5 selections of genes before SVM-R,
namely $S_1,\ldots,S_5$. In addition, we also have the 82 genes
selected from all the subjects previously. Table \ref{t1}
presents the pairwise intersections of these 6 sets with each
other. The union of the 5 selections includes 211 genes, which covers
77 genes in the 82 genes. 73 out of 211 genes have frequency more than
1 where 63 of them appear in the 82 genes. After implementing SVM-R,
the union of genes is reduced to 98 genes. \textit{Figure 4} displays
the histogram of these 211 genes colored by the frequency after SVM-R
runs for $d=1/5$. The pink region denotes the parts further eliminated by SVM-R, which is consistent with the DCOV selection in that SVM-R further rules out the genes with low frequency in the union.                                                                                                                           \begin{table}
\centering
\begin{tabular}{c|ccccc}
\toprule
   & $S_1$ & $S_2$ & $S_3$ & $S_4$ & $S_5$  \\
    \hline
  $S_1$ & 53 &  &  &  &  \\
  $S_2$ & 16 & 77 &  &  &  \\
  $S_3$ & 23 & 21 & 87 & &  \\
  $S_4$ & 18 & 16 & 15 & 33 &  \\
  $S_5$ & 27 & 30 & 31 & 21 & 94 \\
  \hline
  82 genes & 38 & 38 & 44 & 28 & 50 \\
\bottomrule
\end{tabular}
\caption{Pairwise intersections of $S_1,\ldots,S_5$ and the 82 genes. The diagonal numbers are the numbers of selected genes in each $S_i$.}
\label{t1}
\end{table}

\begin{figure}
\begin{center}
  \includegraphics[width=10cm]{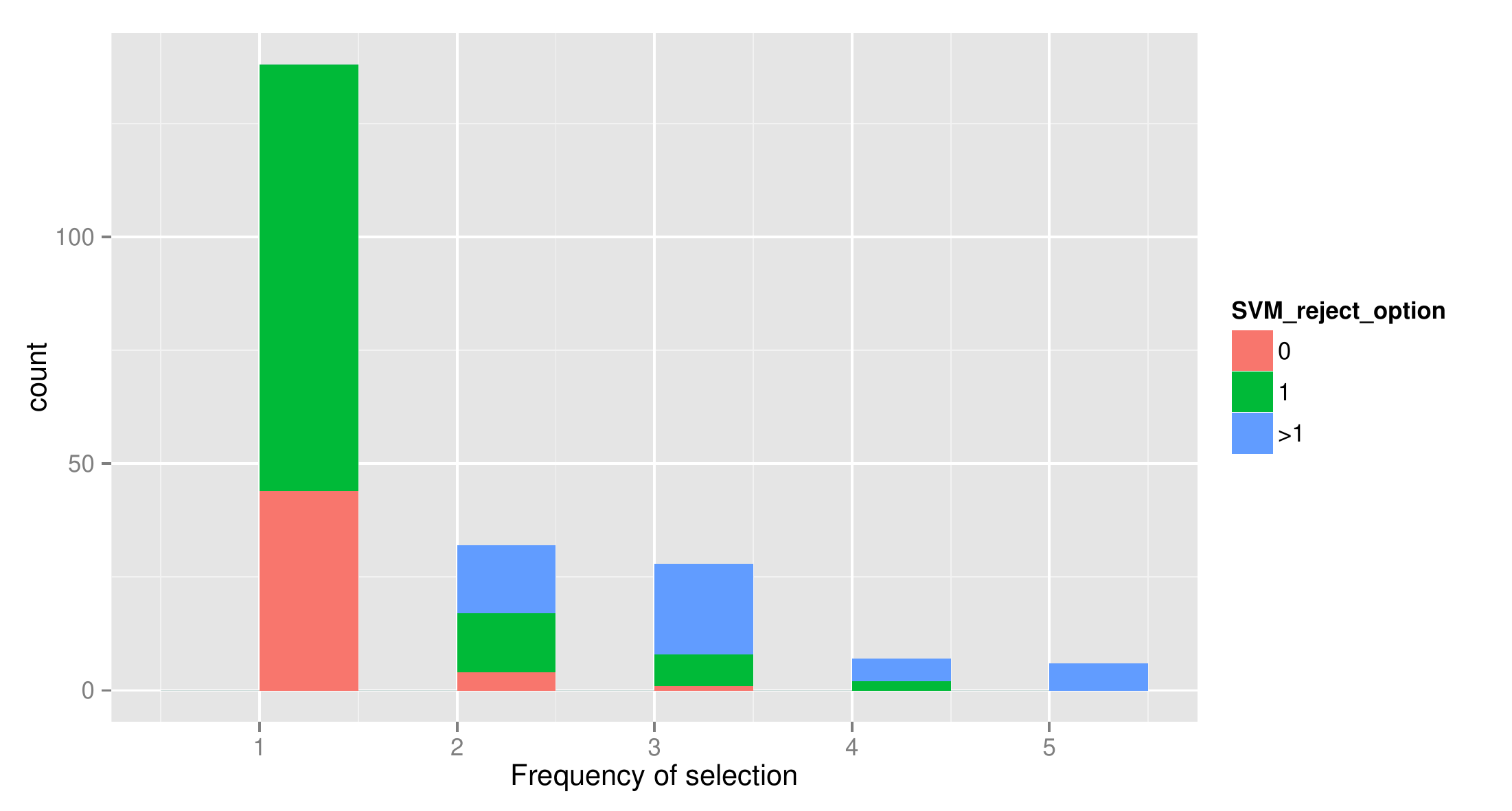}
  \caption{Frequency for the union of $S1,\cdots, S5$, colored by
    frequencies after SVM-R for $d=1/5$.}
\end{center}
\end{figure}

\subsection*{Multiple cross validation}
In order to consider uncertainty in variable selection and model
building due to different partitions of the dataset, we further
extended the five fold cross validation to multiple cross
validation (MCV) and assessed the prediction power through the
following procedure. The results were summarized in the upper part of
Table \ref{t2}.
\begin{enumerate}
\item Randomly partition 279 samples into a 1/5 tuning set, a $2/3\times 4/5 = 8/15$ training set and a $1/3\times 4/5 = 4/15$ testing set.
\item Select genes from the 12024 genes using the proposed method on the training set.
\item Build SVM-R model using the selected genes and the two cancer status variables based on the training set.
\item Use the tuning set to choose the tuning parameter for SVM-R with
  mean $l-$loss as the criteria.
\item Use the model with chosen parameter to predict labels for the testing set.
\item Repeat $1.-5.$ for 50 times.
\end{enumerate}

\begin{table}
  \centering
  \begin{tabular}{l | c |  c c c c c }
    \toprule
    &   & num of reps & mean training & mean testing & mean
    num training & mean num test \\
    & $d$ & with decision &  accuracy(std) & accuracy(std) &with
    decision(std) & with decision(std) \\
    \hline
    \multirow{3}{*}{original}
    &$1/3$& 50 & 0.8319(0.0914) & 0.6943(0.0544) & 101.9400(27.8043) & 49.1800(15.4295) \\
    &$1/4$& 43 & 0.9371(0.0336) & 0.7807(0.1250) & 43.0698(27.0338)  & 20.1860(12.9638)\\
    &$1/5$& 37 & 0.9420(0.0358) & 0.8215(0.1460) & 24.4595(21.3874) & 11.4865(10.0626)\\
    \hline
    \multirow{3}{*}{permute}
    &$1/3$& 49 & 0.7984(0.0078) & 0.6910(0.0426) & 112.1837(26.4352) & 55.5918(13.9954) \\
    &$1/4$& 28 & 0.9225(0.0023) & 0.6867(0.0810) & 56.9643(21.4346)  & 25.2857(10.4132)\\
    &$1/5$& 9  & 0.9686(0.0015) & 0.7071(0.1322) & 53.4444(28.3333) & 24.5556(13.0682)\\
    \bottomrule
  \end{tabular}
\caption{Results for the 50 individual replications for $d = 1/3, 1/4$
  and $1/5$. The upper and lower part are results for the original and
  permuted data respectively. The third column shows the number of
  replications out of 50 with at least one definite decision made on
  the testing set. The fourth and fifth columns of the
  table conclude the mean training and testing accuracies with
  standard deviation in the parenthesis respectively
  restricted to the repetitions with decision made. The last two
  columns display the mean and standard deviation for the number of
  patients assigned decisions for the training and testing sets
  respectively given the replications with decision made.}
\label{t2}
\end{table}

To understand more about the 50 models, we further explored the prediction
results for $d = 1/5$. The prediction labels from the 50 models were
aggregated, following the idea of ensemble methods. The result for
each individual was recorded in a vector of three frequencies, namely
the frequency of being classified as sensitive subjects, the frequency
of obtaining a withholding and the frequency of being assigned to be
resistant out of the 50 models. Let $(s_i, w_i, r_i)$ be the vector for the $i$th patient.

A finer analysis was conducted by looking at the strength of being
sensitive or resistant according to $(s_i, w_i, r_i).$ A voting score
$v_i$ was defined as $(s_i - r_i) / w_i$. Hence, a positive $v_i$
indicated a tendency of being sensitive whereas a negative $v_i$
suggested more possibility of being resistant.

Table \ref{t3} (upper part) partitions the voting scores
into 5 intervals and describes the distribution of $v_i$'s as well as
the proportion of sensitive subjects within each range of $v_i$,
compared to the overall proportion of sensitive patients,
i.e. $191/279$, in the population. It turned out that the trend of
being sensitive weakened monotonically as the voting score decreased. The
stratification specified a subgroup of 15 patients, who possessed the
greatest voting scores, with very high accuracy to be
chemosensitive. The next highest voting score subgroup of 47 subjects
also showed strong confidence of being sensitive compared to the
sample proportion. The conclusion from partitioning the voting scores
was conservative but led to more convincing and steady classification results.

\begin{table}[h]
  \centering
  \begin{tabular}{l c c c c c c c}
    \toprule
    & & &\multicolumn{4}{c}{voting score}\\
    \cmidrule(l){4-8}
    & & &$(-0.1,0]$ & $(0,0.1]$ & $(0.1,0.2]$ & $(0.2,0.4]$ & $(0.4,1.5]$ \\
    \hline
    \multirow{2}{*}{original}
    & \vline &frequency& 76 & 74 & 67 & 47 & 15\\
    &\vline &proportion& 0.5658 & 0.6486 & 0.7164 & 0.8085 & 0.9333 \\
    \hline
    \multirow{2}{*}{permuted}
    &\vline &frequency& 145 & 67 & 43 & 24 & 0\\
    &\vline &proportion& 0.6759 & 0.6866 & 0.7209 & 0.6667 & NA \\
    \bottomrule
 \end{tabular}
\caption{
Frequency of voting score $v_i$'s and proportion of sensitive subjects in each subinterval for $d = 1/5$. The upper and lower parts correspond to the original and permuted data respectively.}
\label{t3}
\end{table}

Each replication of the 50 multiple cross validations gave rise to a different collection of selected genes. This issue is common to selecting
variables from a humongous number of candidates, in the not-low-hanging-fruit situation. The union of the 50 gene selections before SVM-R consisted of 1245 genes
and included all the 82 genes discussed previously. 34 out of 1245
genes were chosen at least 10 times, where 33 of them appeared in the
82 genes, but very few appeared in more than 25 runs. The $l_1$
penalty provided additional elimination, and for $d = 1/5$, 498 out of 1245 genes remained after SVM-R runs. \textit{Figure 5} displays the histogram for the 1245 genes before SVM-R. We distinguished their frequency after SVM-R by different colors. It is shown that a large number of genes with low frequency are further deleted by SVM-R model, i.e. pink color.

\begin{figure}
\begin{center}
  \includegraphics[width=12cm]{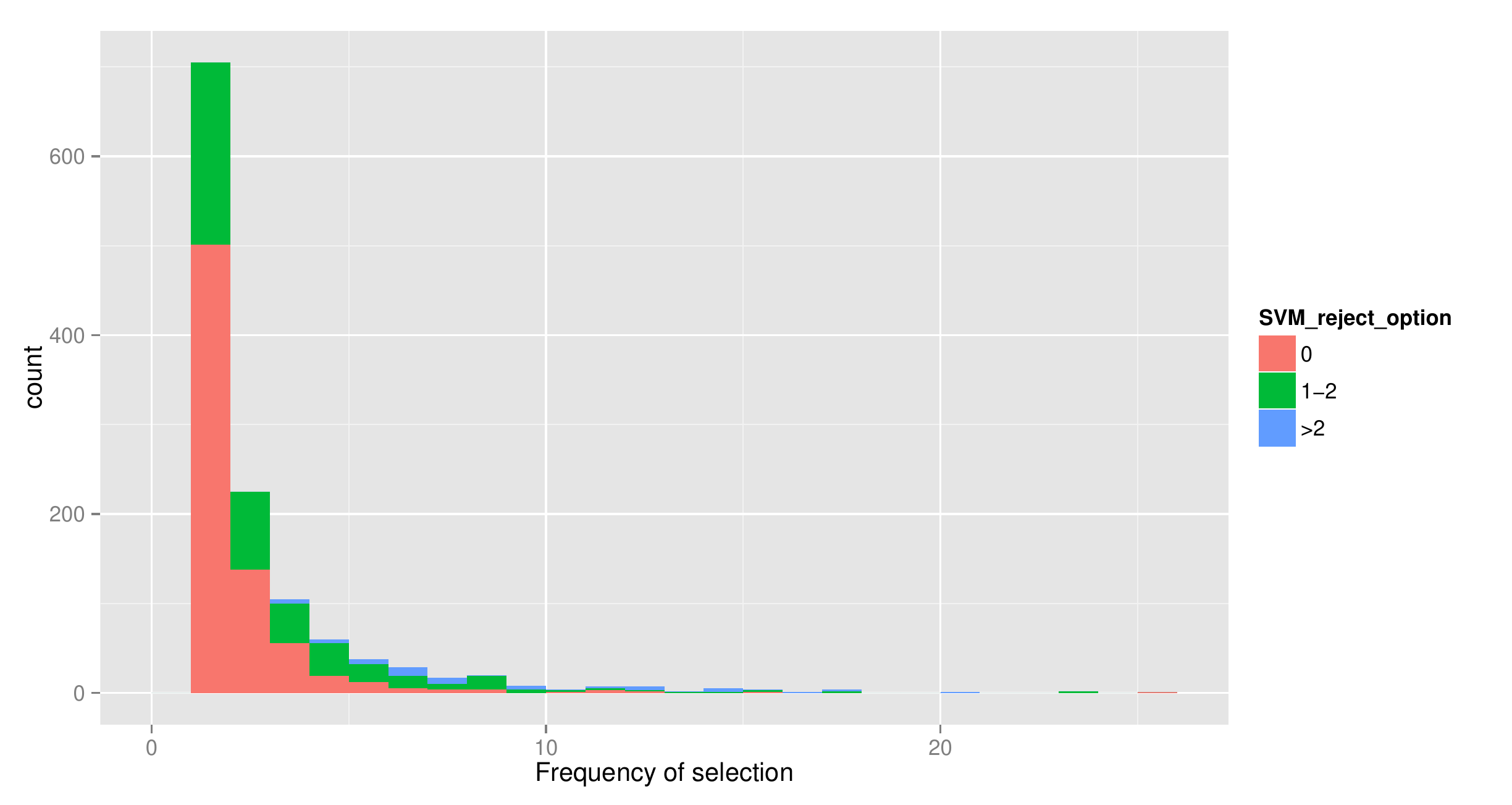}
  \caption{Frequency for 1245 genes being selected by DCOV method,
    colored by frequencies after SVM-R for $d=1/5$.}
\end{center}
\end{figure}

\subsection*{Permutation of the response}
Our method involved several components, including variable selection,
SVM-R, MCV and the voting score, which were interacting with each
other, and led to 15 patients with over $93\%$ accuracy to be sensitive for $d = 1/5$. To further understand the mechanism and to demonstrate that the outcomes were not produced by noises, we randomly permuted the response and went through the whole procedure to compare the results with those for the original data.

It followed that the DCOV method selected genes spuriously correlated
with the permuted response based on the training data in each
replication of the 50 MCVs. The maximum distance correlation value of
the selected genes in each repetition was very close to that for the
original data. The highly correlated genes appeared due to the high
dimensionality of over 12000 genes and less than 200 training
samples.

However, the MCV step played the role of a safeguard against
the fake signals. As Table \ref{t2} depicted, the mean
training accuracies for the original and permuted data showed similar
behavior for the original and permuted data, meaning that the selected
genes were indeed important for the training data. Thus, the chosen
genes in the permutation set should provide little prediction power for the tuning and
testing data. Hence, the validation sets
selected large tuning parameter values driving all the patients with
no decision for many of the 50 replications for $d=1/4$ and
$1/5$. This did not happen for $d=1/3$ since the sample ratio
$191/279$ is slightly greater than $1/3$. For the rest of the replications
with decision making, the mean testing accuracies for the permuted data
remained at the level of the sample proportion of sensitive subjects
for all three values of $d$, which deviated much from the increasing
pattern in the original data.

These suggested that the MCV procedure was able to provide double
fail-secure for fake signals. On one hand, SVM-R placed a cap on the
conditional probability of misclassfication and eliminated the
replications where the selected genes could not produce results
achieving the specified confidence on the validation set. On the other hand, the
mean testing accuracies on the replications passing through the
safeguard of tuning sets would be no better than assigning everyone to
the sensitive class when there was no real signal.

Furthermore, the poor prediction performance of the 50 individual
models ended up with unsurprisingly disappointing voting score results
for the permuted data, as shown in the lower part of Table
\ref{t3}. Many of the patients obtained a relatively small
value of the voting score and nobody got a score in the range where
the original data had the highest accuracy, meaning that the confidence was quite
low. Moreover, the stratified ratios of sensitive subjects for different
ranges of the voting scores did not show anything insightful other than
being around the sample proportion.

\section*{Discussion}
The paper introduced a new variable selection procedure based on the property of distance covariance and demonstrated the application through two examples. The small round blue cell tumors data played a role of a toy example to show that the performance of the proposed method worked well in easy cases. The TCGA ovarian cancer data, however, were much more challenging to deal with due to the humongous number of variables and very limited sample size. The uncertainty of variable selection was discussed through gene selection results using random subsets of the data. The support vector machine with reject option was used to withhold decision for subjects who were difficult to classify. An ensemble method of combining models built on random subsets of the data was implemented to assess the prediction performance. Although we had applied these tools (DCOV, SVM-R, MCV) to biomedical data in the paper, we argue that they are quite portable across
disciplines.

As shown in Table \ref{t3}, a small portion of the model building population got classified for $d = 1/5$. Is it worthwhile to attempt the classification in such cases? It depends on the application, for example differential costs of two types of misclassification, and subjective considerations including quality of life influenced by the treatment, therapy expense and expected survival time.

Both the analysis of five fold cross validation and multiple cross validation showed the uncertainty of gene selection results based on different subsets of the data. The large number of variables that appeared only in a small number of runs suggested noises in the data and the difficulty caused by limited training sample size in the high dimensional scenario. It could also suggest the conundrum that the ``true'' model consists of a large number of variables with modest effects of which different subsets gives rise to roughly equal prediction ability. Options for further study in this and other difficult problems include allowing the DCOV stopping criteria to be modified by some amount $\epsilon$, and allowing the greedy variable selection algorithm to be doubly greedy by testing the next best $m$ of the remaining variables rather than just the next variable. It remains to obtain theoretical results to guide exploration in alternate scenarios.

The analysis of random permutation on the response served as both a
validation of our results and a discussion of what one is likely to
obtain without any true signal. If someone started with an entirely
different data set having the same proportions for the two classes with
that in the original data but no real signal at all, as what one might
get from scrambling, and went through every step, and finally obtained
a subgroup of patients with large voting score values, the result was no better than just guessing that everyone was sensitive. This experiment was also a cautionary tale that if one had not held out validation sets, the analyst could be easily fooled by spurious correlated variables and perfect training accuracy. Our proposed multiple cross validation and analysis through the voting scores provided protection against finding fake signals.


\begin{thebibliography}{9}
\bibitem{fan2008sure}
J.~Fan and J.~Lv.
\newblock Sure independence screening for ultrahigh dimensional feature space.
\newblock {\em Journal of the Royal Statistical Society: Series B (Statistical
  Methodology)}, 70(5):849--911, 2008.

\bibitem{szekely2007measuring}
G.J. Sz{\'e}kely, M.L. Rizzo, and N.K. Bakirov.
\newblock Measuring and testing dependence by correlation of distances.
\newblock {\em The Annals of Statistics}, 35(6):2769--2794, 2007.

\bibitem{li2012feature}
R.~Li, W.~Zhong, and L.~Zhu.
\newblock Feature screening via distance correlation.
\newblock {\em Journal of the American Statistical Association}, 107(499):1129
  -- 1139, 2012.

\bibitem{kong2012using}
Jing Kong, Barbara~EK Klein, Ronald Klein, Kristine~E Lee, and Grace Wahba.
\newblock Using distance correlation and ss-anova to assess associations of
  familial relationships, lifestyle factors, diseases, and mortality.
\newblock {\em Proceedings of the National Academy of Sciences},
  109(50):20352--20357, 2012.

\bibitem{kosorok2009discussion}
M.R. Kosorok.
\newblock Discussion of: brownian distance covariance.
\newblock {\em The Annals of Applied Statistics}, 3(4):1270--1278, 2009.

\bibitem{khan2001classification}
J.~Khan, J.S. Wei, M.~Ringner, L.H. Saal, M.~Ladanyi, F.~Westermann,
  F.~Berthold, M.~Schwab, C.R. Antonescu, C.~Peterson, et~al.
\newblock Classification and diagnostic prediction of cancers using gene
  expression profiling and artificial neural networks.
\newblock {\em Nature medicine}, 7(6):673--679, 2001.

\bibitem{lee2004multicategory}
Yoonkyung Lee, Yi~Lin, and Grace Wahba.
\newblock Multicategory support vector machines: Theory and application to the
  classification of microarray data and satellite radiance data.
\newblock {\em Journal of the American Statistical Association},
  99(465):67--81, 2004.

\bibitem{website}
National Cancer~Institute website.
\newblock http://www.cancer.gov/cancertopics/types/ovarian.
\newblock {\em Accessed Aug 28, 2013}.

\bibitem{bell2011integrated}
D~Bell, A~Berchuck, M~Birrer, J~Chien, DW~Cramer, F~Dao, R~Dhir, P~DiSaia,
  H~Gabra, P~Glenn, et~al.
\newblock Integrated genomic analyses of ovarian carcinoma.
\newblock 2011.

\bibitem{selvanayagam2004prediction}
Zachariah~E Selvanayagam, Tak~Hong Cheung, Nien Wei, Ragini Vittal, Keith~Wing
  Kit~Lo, Winnie Yeo, Tsunekazu Kita, Roald Ravatn, Tony~Kwok Hung~Chung,
  Yick~Fu Wong, et~al.
\newblock Prediction of chemotherapeutic response in ovarian cancer with dna
  microarray expression profiling.
\newblock {\em Cancer genetics and cytogenetics}, 154(1):63--66, 2004.

\bibitem{rho2008insulin}
Seung~Bae Rho, Seung~Myung Dong, Sokbom Kang, Sang-Soo Seo, Chong~Woo Yoo,
  Dong~Ock Lee, Jong~Soo Woo, and Sang-Yoon Park.
\newblock Insulin-like growth factor-binding protein-5 (igfbp-5) acts as a
  tumor suppressor by inhibiting angiogenesis.
\newblock {\em Carcinogenesis}, 29(11):2106--2111, 2008.

\bibitem{mehlmann2004gs}
Lisa~M Mehlmann, Yoshinaga Saeki, Shigeru Tanaka, Thomas~J Brennan, Alexei~V
  Evsikov, Frank~L Pendola, Barbara~B Knowles, John~J Eppig, and Laurinda~A
  Jaffe.
\newblock The gs-linked receptor gpr3 maintains meiotic arrest in mammalian
  oocytes.
\newblock {\em Science}, 306(5703):1947--1950, 2004.

\bibitem{ledent2005premature}
Catherine Ledent, Isabelle Demeestere, David Blum, Julien Petermans, Tuula
  H{\"a}m{\"a}l{\"a}inen, Guillaume Smits, and Gilbert Vassart.
\newblock Premature ovarian aging in mice deficient for gpr3.
\newblock {\em Proceedings of the National Academy of Sciences of the United
  States of America}, 102(25):8922--8926, 2005.

\bibitem{basu2009nanoparticle}
Sudipta Basu, Rania Harfouche, Shivani Soni, Geetanjali Chimote, Raghunath~A
  Mashelkar, and Shiladitya Sengupta.
\newblock Nanoparticle-mediated targeting of mapk signaling predisposes tumor
  to chemotherapy.
\newblock {\em Proceedings of the National Academy of Sciences},
  106(19):7957--7961, 2009.

\bibitem{bast2010personalizing}
Robert~C. Bast and Gordon~B. Mills.
\newblock Personalizing therapy for ovarian cancer: Brcaness and beyond.
\newblock {\em Journal of Clinical Oncology}, 28(22):3545--3548, 2010.

\bibitem{thiele2011expression}
Sylvia Thiele, Martina Rauner, Claudia Goettsch, Tilman~D Rachner, Peggy Benad,
  Susanne Fuessel, Kati Erdmann, Christine Hamann, Gustavo~B Baretton,
  Manfred~P Wirth, et~al.
\newblock Expression profile of wnt molecules in prostate cancer and its
  regulation by aminobisphosphonates.
\newblock {\em Journal of cellular biochemistry}, 112(6):1593--1600, 2011.

\bibitem{bitler2011wnt5a}
Benjamin~G Bitler, Jasmine~P Nicodemus, Hua Li, Qi~Cai, Hong Wu, Xiang Hua,
  Tianyu Li, Michael~J Birrer, Andrew~K Godwin, Paul Cairns, et~al.
\newblock Wnt5a suppresses epithelial ovarian cancer by promoting cellular
  senescence.
\newblock {\em Cancer research}, 71(19):6184--6194, 2011.

\bibitem{cole2010inhibition}
Claire Cole, Sin Lau, Alison Backen, Andrew Clamp, Graham Rushton, Caroline
  Dive, Cassandra Hodgkinson, Rhona McVey, Henry Kitchener, and Gordon~C
  Jayson.
\newblock Inhibition of fgfr2 and fgfr1 increases cisplatin sensitivity in
  ovarian cancer.
\newblock {\em Cancer biology \& therapy}, 10(5):495--504, 2010.

\bibitem{yin2011genetic}
Jikai Yin, Karen Lu, Jie Lin, Lei Wu, Michelle~AT Hildebrandt, David~W Chang,
  Larissa Meyer, Xifeng Wu, and Dong Liang.
\newblock Genetic variants in tgf-$\beta$ pathway are associated with ovarian
  cancer risk.
\newblock {\em PloS one}, 6(9):e25559, 2011.

\bibitem{gu2007gss}
Chong Gu.
\newblock gss: General smoothing splines.
\newblock {\em R package version 2.1-2}, 2007.

\bibitem{wegkamp2011support}
Marten Wegkamp and Ming Yuan.
\newblock Support vector machines with a reject option.
\newblock {\em Bernoulli}, 17(4):1368--1385, 2011.

\bibitem{bartlett2008classification}
Peter~L. Bartlett and Marten~H. Wegkamp.
\newblock Classification with a reject option using a hinge loss.
\newblock {\em The Journal of Machine Learning Research}, 9:1823--1840, 2008.

\end{thebibliography}
\end{document}